\documentclass[submission,copyright,creativecommons]{eptcs}
\providecommand{\event}{ICE 2018} % Name of the event you are submitting to
\usepackage{breakurl}

\usepackage{amsmath,amsthm,amssymb}
\usepackage{pstricks,pstricks-add,pst-node,pst-eucl,pst-tree}
\usepackage{graphics}

\newcommand{\N}{\mathbb{N}}

\newcommand{\taille}[1]{\lvert #1 \rvert}

\newcommand{\GRrondn}{G_n}
\newcommand{\fonction}[6]{\ifnum #1=0
\begin{array}[t]{rcl}
{#3} & \longrightarrow & {#4} \\
{#5} & \longmapsto & {#6} \end{array}
\else
\begin{array}[t]{lrcl}
{#2} : & {#3} & \longrightarrow & {#4} \\
    & {#5} & \longmapsto & {#6} \end{array}\fi}

\newcommand{\distinguer}[1]{{#1}_{K}}
\newcommand{\ssterme}[2]{{#1}{\downarrow #2}}
\newcommand{\sub}[2]{[#1\leftarrow#2]}
\newcommand{\arbre}[2]{\tree(#1,#2)}

\newcommand{\rond}[1]{\mathcal{#1}}

\newcommand{\GTrs}[1]{\mathbf{#1}}

\newcommand{\Rrond}{\mathbf{R}}

\newcommand{\foata}[1]{\lceil #1 \rceil_{\Foata}}

\newcommand{\residual}[2]{\Res(#1,#2)}

\DeclareMathOperator{\Fr}{Fr}
\DeclareMathOperator{\minimum}{min}
\DeclareMathOperator{\maximum}{max}
\DeclareMathOperator{\U}{Unf}
\DeclareMathOperator{\dec}{dec}
\DeclareMathOperator{\Accs}{Reach}
\DeclareMathOperator{\REC}{Rec}

\DeclareMathOperator{\Reg}{Reg}
\DeclareMathOperator{\Res}{Res}

\DeclareMathOperator{\FOAccs}{\FO[\Accs]}

\DeclareMathOperator{\pos}{Dom}
\DeclareMathOperator{\FO}{FO}
\DeclareMathOperator{\graph}{Gr}

\DeclareMathOperator{\tree}{Tree}

\DeclareMathOperator{\StructEv}{\mathcal{E}\mathcal{S}}
\DeclareMathOperator{\premiere}{prime}

\DeclareMathOperator{\goto}{goto}
\DeclareMathOperator{\Decr}{Decr}
\DeclareMathOperator{\Incr}{Incr}
\DeclareMathOperator{\Foata}{\mathcal{\mathbf{F}}}

\DeclareMathOperator{\neutral}{neutral}

\theoremstyle{plain}
\newtheorem{theo}{Theorem}[section]
\theoremstyle{plain}
\newtheorem{prop}[theo]{Proposition}
\theoremstyle{plain}
\newtheorem{lemme}[theo]{Lemma}
\theoremstyle{plain}
\newtheorem{corollaire}[theo]{Corollary}
\theoremstyle{plain}
\newtheorem{affirmation}{Claim}
\theoremstyle{definition}
\newtheorem{df}[theo]{Definition}
\theoremstyle{definition}
\newtheorem{exemple}[theo]{Example}
\theoremstyle{remark}
\newtheorem{remarque}[theo]{Remark}

\title{Unfolding of Finite Concurrent Automata}
\author{Alexandre Mansard
\institute{LIM - University of La R\'eunion\\ France}
\email{alexandre.mansard@univ-reunion.fr}}
\def\titlerunning{Unfolding of Finite Concurrent Automata}
\def\authorrunning{A. Mansard}

\begin{document}
\maketitle

\begin{abstract}
We consider recognizable trace rewriting systems with level-regular contexts (RTL). A trace language is level-regular if the set of Foata normal forms of its elements is regular. We prove that the  rewriting graph of a RTL is word-automatic. Thus its first-order theory is decidable. Then, we prove that the concurrent unfolding of a finite concurrent automaton with the reachability relation is a RTL graph. It follows that the first-order theory with the reachability predicate ($\FOAccs$ theory) of such an unfolding is decidable. It is known that this property holds also for the ground term rewriting graphs. We provide examples of finite concurrent automata of which the concurrent unfoldings fail to be ground term rewriting graphs. The infinite grid tree (for each vertex of an infinite grid, there is an edge from  this vertex to the origin of a copy of the infinite grid) is such an unfolding. We prove that the infinite grid tree is not a ground term rewriting graph. We have thus obtained a new class of graphs for with a decidable $\FOAccs$ theory.

\end{abstract}

\section{Introduction}

A challenging problem in automatic verification consists in determining (or in extending) classes of infinite graphs having a decidable theory in a given logic. A first technique consists in considering some judicious graph transformations, as for example unfolding (that preserves decidability of monadic second-order logic) or logical interpretations. The pushdown hierarchy  \cite{DBLP:conf/mfcs/Caucal02} is a hierarchy of  decidable graphs of monadic second-order theory. Starting from finite graphs, each level consists of the monadic interpretations of the unfoldings of  lower levels. The tree-automatic hierarchy \cite{DBLP:journals/lmcs/ColcombetL07} is a hierarchy of graphs of decidable first-order (FO) theory: each level consists of finite set interpretations of the corresponding level of the pushdown hierarchy. A second technique is to consider graphs whose vertex set and relations are recognizable by automata whose recognized languages form a Boolean algebra. For instance, it is the case of word-automatic graphs or more generally tree-automatic graphs, \textit{i.e} graphs whose vertex set can be encoded by a regular tree language and each relation recognized by a synchronized tree transducer. It turns out the first level of the tree-automatic hierarchy consists of tree-automatic graphs. Lastly, rewriting systems also allow to define interesting graph classes. Graphs at the first level of the pushdown hierarchy  are the suffix rewriting graphs of recognizable word rewriting systems \cite{DBLP:conf/caap/Caucal90}. Ground term rewriting graphs (GTR graphs) with the reachability relation are tree-automatic and thus the first-order theory with the reachability predicate ($\FOAccs$) of a GTR graph is decidable \cite{DBLP:conf/lics/DauchetT90}.

Since its monadic second-order theory is not decidable, the infinite grid does not belong to the pushdown hierarchy and is therefore not the unfolding of a finite graph. Nevertheless, as a GTR graph, the infinite grid has a decidable  $\FOAccs$ theory. In fact, even the theory of the infinite grid in first-order logic extended by the operator of transitive closure for first-order definable relations remains decidable \cite{DBLP:conf/lics/WohrleT04}. But consider now  the infinite grid tree: from each vertex of an infinite grid, there is an edge (labelled by a new symbol) to the origin of a copy of the infinite grid. We will prove that this simple graph (it is just the configuration graph of a system with 2 counters that we can independently incremente and simultanely reset) has the $\FOAccs$ theory decidable but is not a GTR graph. In fact, we are interested in considering, more generally, a class of graphs that  model concurrent system computations. For such a system, sequential and parallel computations are possible. To that end, we will consider Mazurkiewicz traces: if the dependency is total, then a trace reduces to a string that describes sequential computation while independence between some letters bring the possibility to describe parallel computation.

For a recognizable trace rewriting system, that is a finite set of rules of the form $\rond{U}\cdot(\rond{V}\xrightarrow{\lambda} \rond{W})$ where $\rond{U}$, $\rond{V}$, $\rond{W}$ are recognizable trace languages, $\lambda$ a label, consider then its rewriting graph: the set of edges of the form $ts\xrightarrow{\lambda}ts'$ such that there exists a rewriting rule  $\rond{U}\cdot(\rond{V}\xrightarrow{\lambda} \rond{W})$ with $t\in\rond{U}$, $s\in\rond{V}$, $s'\in\rond{W}$. If all letters are dependent, then such a graph is at the first level of the pushdown hierarchy since it is the suffix rewriting graph of a recognizable word rewriting system, and if no distinct letters are dependent, then it is the configuration graph of a vector addition system. In any case, we will prove in Section 3 that such a graph is word-automatic, even with level-regular contexts: a trace language is level-regular if the set of Foata normal forms of its elements is regular. Since the set of Foata normal forms is regular, every recognizable trace language is level-regular. But, for example, if $a$ and $b$ are two independent letters, the trace language $[(ab)^*]$ is level-regular but not recognizable. The FO theory of the rewriting graph of a recognizable trace rewriting system with level-regular contexts (RTL graph) is thus decidable. We also prove that, in general, its $\FOAccs$ theory is not decidable. Otherwise, we could decide the halting problem for 2-counters Minsky machine. In Section 4, we prove that the concurrent unfolding of a finite concurrent automaton  has the $\FOAccs$ theory decidable, by showing that such a graph with the reachability relation is a RTL graph. This extends a theorem of Madhusudan \cite{DBLP:conf/lics/Madhusudan03} on decidability of FO theory of regular trace event structures \cite{Thiagarajan_1996}. We will observe that the infinite grid and the infinite grid tree are the concurrent unfoldings of finite concurrent automata. In Section 5, we define the tree of a graph and we prove that if it is a GTR graph, then it is finitely decomposable by size. The latter implies it is at the first level of the pushdown hierarchy. We deduce that the infinite grid tree is not a GTR graph.

\section{Preliminaries}Before presenting the rewriting graphs of recognizable trace rewriting systems, we recall some basic definitions about graphs, logics, automata and traces.

\medskip
Let $\Sigma$ be a finite alphabet and $\Sigma^*$ be the free monoid of words over $\Sigma$.

\subsection{Graphs}
A \emph{$\Sigma$-graph} $G$ is a subset of $V\times \Sigma \times V$ where $V$ is a set. An element $(p,a,q)\in G$ is an edge labelled by $a$ from source $p$ to target $q$. The notation $p\xrightarrow[G]{a}q$ (or $p\xrightarrow[]{a}q$ when $G$ is understood) means $(p,a,q)\in G$. The \emph{vertex} set of $G$ is $V_G=\{p\in V\mid \exists q\ (p\xrightarrow[G]{a}q \lor q\xrightarrow[G]{a}p)\}$. 

The graph $G$ is \emph{deterministic} if for every $a\in \Sigma$, if $(p\xrightarrow[G]{a}q$ and $p\xrightarrow[G]{a}q')$ then $q=q'$. 

A \emph{path} in $G$ between vertices $p$ an $q$, labelled by a word $u=a_1\dots a_k$ is a finite sequence of the form $p\xrightarrow[G]{a_1}p_1,\dots,p_{k-1}\xrightarrow[G]{a_k}q$. We denote by $p\xrightarrow[G]{a_1\dots a_k}q$ the existence of such a path. A \emph{loop} in $G$ is a path of the form $p\xrightarrow[G]{a}p$, where $a\in \Sigma$. If $G$ is finite, then for $p,q\in V_G$, the $\Sigma$-word language $L_{p,q}:=\{u\in\Sigma^*\mid p\xrightarrow[G]{u}q\}$ is regular. We write $p\xrightarrow[G]{*}q$ if there exists a word $u\in \Sigma^*$ such that $p\xrightarrow[G]{u}q$. Denote by $G_*$ the  $\Sigma\ \dot\cup\ \{*\}$-graph defined by $G_*=G\cup\{(p,*,q)\mid p\xrightarrow[G]{*}q\}$. The graph $G_*$ is obtained from $G$ by adding the reachability relation.
 
 An isomorphism $f$ from $(G,P)$ onto $(H,Q)$, where $G$ and $H$ are $\Sigma$-graphs and $P$ and $Q$ are subsets of $V_G$ and $V_H$ respectively, is a bijection from $V_G$ to $V_H$ such that $$f(P)=Q\ \text{ and }\ (p\xrightarrow[G]{a}q\Longleftrightarrow f(p)\xrightarrow[H]{a}f(q)).$$
\subsection{Logics}A $\Sigma$-graph $G$ is a relational structure over the binary signature $\Sigma$. The \emph{first-order} ($\FO$) theory of $G$ is defined as usual (see \cite{ebbinghaus1996mathematical}). The FO theory of $G_*$ will be refered to as the $\FOAccs$ theory of $G$.

\subsection{Automata}A $\Sigma$-automaton is a triple $\rond{A}=(G,i,F)$ where $G$ is a $\Sigma$-graph, $i\in V_G$ is an initial state and $F\subseteq V_G$ is the set of final states. The $\Sigma$-word language recognized by $\rond{A}$ is $L(\rond{A})=\{u\in \Sigma^*\mid \exists f\in F\ i\xrightarrow[G]{u}f\}$. A $\Sigma$-word language is \emph{regular} if it is recognized by a finite $\Sigma$-automaton. The class of regular $\Sigma$-word languages is a Boolean algebra and is denoted by $\Reg(\Sigma^*)$.

\subsection{Traces}
\subsubsection{Generalities}
A \emph{dependence relation} $D$ is a reflexive and symmetric binary relation on $\Sigma$. The pair $(\Sigma,D)$ is called a \emph{dependence alphabet}. The complement of $D$ is the \emph{independence relation} $I:=\Sigma^2\backslash D$. The $(\Sigma,D)$-trace equivalence $\equiv_{D}$ is the least congruence on $\Sigma^*$ such that $(a,b)\in I\Rightarrow ab\equiv_{D} ba$.
The $(\Sigma,D)$-trace of a word $w\in\Sigma^*$ is its $\equiv_{D}$-equivalence class and is denoted $[w]$. The quotient monoid $\Sigma^*/\equiv_{D}$ is called the \emph{trace monoid} of the dependence alphabet $(\Sigma,D)$ and is denoted by $M(\Sigma,D)$. Note that in case of $D=\Sigma^2$, the trace monoid $M(\Sigma,D)$ coincides to the free monoid $\Sigma^*$. 

The prefix binary relation $\sqsubseteq$ on $M(\Sigma,D)$ defined by $t\sqsubseteq t'$ if and only if there exists $s\in M(\Sigma,D)$ such that $ts=t'$ is a partial ordering.

Consider the finite alphabet  $I_D:=\{A\subseteq \Sigma \mid \forall a_1\neq a_2\in A\  (a_1,a_2)\in I\}$ and denote by $\Pi_{I_D}:I_D^*\rightarrow M(\Sigma,D)$ the canonical morphism  defined by $\Pi_{I_D}(\varnothing)=[\varepsilon]$ and $\Pi_{I_D}(\{a_1,\cdots,a_n\})=[a_1\dots a_n]$ ($n\geqslant 0$). Given $P\subseteq I_D$, we denote by $\Pi_{P}$ the restriction of $\Pi_{I_D}$ to $P^*$. A $P$-word $U$ encodes the trace $\Pi_P(U)$.

Consider the binary relation $\rhd$ on $I_D^{-}:=I_D\setminus \{\varnothing\}$ defined by: $A\rhd B\iff \forall b\in B\ \exists a\in A\ aDb$. Denote by $\Foata\subseteq {I_D^{-}}^*$ the set of $\rhd$-paths.

The surjective morphism $\Pi_{I_D^{-}}$ is not injective. Indeed, suppose $\Sigma=\{a,b\}$ and $aIb$, then $\Pi_{I_D^{-}}(\{a,b\})=\Pi_{I_D^{-}}(\{a\}\{b\})$. The following proposition expresses that each trace is encodable by a unique $I_D^{-}$-word in $\Foata$.
 \begin{prop}[Foata normal form]\label{foata}
   Let $t\in M(\Sigma,D)$. The Foata normal form of $t$, $\foata{t}$, is the unique $I_D^{-}$-word $\foata{t}=A_1\cdots A_p\in \Foata$ ($p\geqslant 0$) such that $\Pi_{I_D^{-}}(A_1\cdots A_p)=t$.
 \end{prop}

 \begin{exemple}\label{exempleFoataNomalForm}Suppose  $\Sigma=\{a,b,c,d\}$ and $aIc$, $bId$, $cId$. The Foata normal form of $t=[acbdab]$ (see Figure \ref{fig:foata}) is $\foata{t}=\{a,c\}\{b,d\}\{a\}\{b\}$.
 \end{exemple}
\begin{figure}
  \centering
  %\scalebox{0.4}{\input{FoataTrace.pstex_t}}
  \scalebox{0.5}{\includegraphics{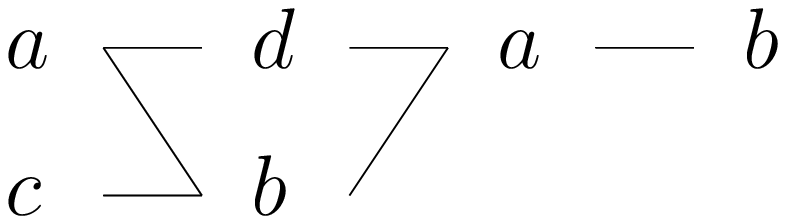}} 
    
    \caption{The Foata normal form of $[acbdab]$ (Example \ref{exempleFoataNomalForm})}
   \label{fig:foata}
  \end{figure}

  \begin{lemme}[Level automata]\label{levelautomata}The set $\Foata$ of Foata normal forms is regular.
\end{lemme}
\begin{proof}

  It is recognized by the following finite $I_D^{-}$-automaton $\rond{A}_{\Foata}$.
   \begin{itemize}
   \item The $I_D^{-}$-graph is given by:
     
     $\bot\xrightarrow{A} A$ : $ A\in I_D^{-}$

     $A\xrightarrow{B} B$ : $A\rhd B$
     
   \item the initial state is $\bot\notin I_D^{-}$
   \item all the states are final (even $\bot$).
   \end{itemize}
 \end{proof}

 \subsubsection{Recognizable trace languages}A $(\Sigma,D)$-\emph{trace language} is a subset of $M(\Sigma,D)$. If $\rond{L}$ is a trace language, then the word language $\cup\rond{L}$ is $\cup\rond{L}=\{w\in\Sigma^*\mid [w]\in \rond{L}\}$. If $L$ is a word language, then $[L]$ is the trace language defined by $[L]:=\{[w]\in M(\Sigma,D)\mid w\in L\}$.

 A trace language $\rond{L}\subseteq M(\Sigma,D)$ is \emph{recognizable} if there exists a finite monoid $N$ and a monoid morphism $\phi : M(\Sigma,D)\rightarrow N$ such that $\rond{L}=\phi^{-1}(\phi(\rond{L}))$. The class of recognizable trace languages is denoted by $\REC(M(\Sigma,D))$.
 \begin{remarque}
In case of $D=\Sigma^2$, $\REC(M(\Sigma,D))=\Reg(\Sigma^*)$.
\end{remarque}

The next proposition recalls the robustness of the class $\REC(M(\Sigma,D))$.

\begin{prop}\label{Boolealgebre}
$\REC(M(\Sigma,D))$ is a Boolean algebra closed under concatenation.
\end{prop}

We give two characterizations of the recognizability of a trace language. 
The residual by $s\in M(\Sigma,D)$ of $\rond{L}\subseteq M(\Sigma,D)$ is $s^{-1}\rond{L}=\{t\in M(\Sigma,D)\mid s\cdot t\in \rond{L}\}$. For example, suppose  $\Sigma=\{a,b\}$ and $aIb$, then consider $\rond{L}=\{[ab],[abaa],[aaa],[aabbb]\}$. The residual by $[ab]$ of $\rond{L}$ is $\{\varepsilon,[aa],[abb]\}$. The recognizability of a trace language $\rond{L}$ is characterized by the finiteness of its set of residuals.
\begin{prop}\label{finituderesidus}
 $\rond{L}\in \REC(M(\Sigma,D))$ if and only if $\{s^{-1}\rond{L}\mid s\in M(\Sigma,D)\}$ is finite.
\end{prop}

Suppose $P$ is a finite alphabet and $\pi:P^*\rightarrow M(\Sigma,D)$ is a surjective morphism. For instance, $P$ could be $\Sigma$, $I_D^{-}$ or $I_D$. If for a trace $t$ we think of $\pi^{-1}(t)$ as the set of its $P$-encodings, the following proposition says that the recognizability of a trace language is equivalent to the regularity of the set of all $P$-encodings of its elements. 

\begin{prop}\label{Recreg}
 $\rond{L}\in \REC(M(\Sigma,D))$ if and only if $\pi^{-1}(\rond{L})$ is regular.

\end{prop}

\section{Recognizable trace rewriting system with level-regular contexts}The trace language $[(ab)^*]$ with $aIb$ is not recognizable since it has an infinite set of residuals. Nevertheless, the set of Foata normal forms of its elements $\{a,b\}^*$ is regular. This suggests to consider a weaker form of recognizability. In this section, we define the notion of level-regularity for trace languages. Then we consider recognizable trace rewriting systems with level-regular contexts and we prove that their rewriting graphs are word-automatic.

Let $(\Sigma,D)$ be a dependence alphabet. In the following, we write $\Pi_{\Foata}$ for the restriction of $\Pi_{I_D^{-}}$ to $\Foata$.
\subsection{Level-regularity}

\begin{df}$\rond{L}\subseteq M(\Sigma,D)$ is level-regular if the word language $\Pi_{\Foata}^{-1}(\rond{L})$ is regular.
\end{df}

 By  Proposition \ref{Recreg} and Lemma \ref{levelautomata}, every recognizable trace language is level-regular. Indeed, $\Pi_{\Foata}^{-1}(\rond{L})=\Pi_{I_D}^{-1}(\rond{L})\cap \Foata$.

The class of level-regular languages is a Boolean algebra but it is not closed under concatenation. Consider for example the concatenation of the two level-regular trace languages $[(ab)^*]$ and $[(bc)^*]$, with $D=\{(a,a),(b,b),(c,c)\}$. The set of Foata normal forms of its elements

\begin{center}
\begin{tabular}{rl}
   &$\Pi_{\Foata}^{-1}([(ab)^*]\cdot [(bc)^*])$\\
   $=$&$\{\{a,b,c\}^k\{b,c\}^*\{b\}^k\mid k\geqslant 0\}\cup \{\{a,b,c\}^k\{a,b\}^*\{b\}^k\mid k\geqslant 0\}$\\
\end{tabular}
\end{center}

is not regular.

\subsection{Trace rewriting system}

Graphs at the first level of the pushdown hierarchy are the suffix rewriting graphs of recognizable word rewriting systems. Such a rewriting system is a finite set of rules of the form $U\cdot (V\xrightarrow{}W)$, where $U$ (the context language), $V$ and $W$ are regular languages. In the following, we consider recognizable trace rewriting systems with  level-regular contexts and recognizable left and right hand sides and we prove that their rewriting graphs are word-automatic by encoding their vertex sets by their Foata normal forms.

\begin{df}A recognizable trace rewriting system with level-regular contexts (RTL) $\Rrond$ on $M(\Sigma,D)$ is a finite set of rules of the form  $$\rond{U}\cdot(\rond{V}\xrightarrow{\lambda} \rond{W})$$
where $\rond{U}$ is level-regular, $\rond{V},\rond{W}\in \REC(M(\Sigma,D))$ and $\lambda\in \Lambda$ a set of labels.

The rewriting graph $\graph_\Rrond$ of the RTL $\Rrond$ is the $\Lambda$-graph on $M(\Sigma,D)$ defined by  $$\graph_\Rrond=\{[uv]\xrightarrow{\lambda}[uw]\mid \exists\ \ \rond{U}\cdot(\rond{V}\xrightarrow{\lambda} \rond{W})\in \Rrond, [u]\in \rond{U},[v]\in \rond{V},[w]\in \rond{W} \}.$$
\end{df}

\begin{exemple}\label{alexbis}
  Suppose $D=\{(a,a),(b,b)\}$ and consider the following RTL:

  $[(a+b)^*]\cdot([\varepsilon]\xrightarrow{a}[a])$

  $[(a+b)^*]\cdot([\varepsilon]\xrightarrow{b}[b])$

  $[(ab)^*]\cdot([\varepsilon]\xrightarrow{f}[\varepsilon])$

  Its rewriting graph is the infinite grid with a loop labelled by $f$ on each vertex of its diagonal (see Figure \ref{fig:grillediag}).
  
\end{exemple}

\begin{exemple}\label{alex}
  Suppose $D=\{(a,a),(b,b),(c,c)\}$ and consider the following RTL $\Rrond$:

  $[(abc)^*]\cdot ([\varepsilon]\xrightarrow{a}[abc])$

  $[(abc)^*(ac)^*]\cdot([b]\xrightarrow{b}[\varepsilon])$

  $[(abc)^*(ac)^*]\cdot([ac]\xrightarrow{c} [\varepsilon])$

The rewriting graph $\graph_\Rrond$  (see Figure \ref{fig:cone})  is not in the pushdown hierarchy because its MSO theory is undecidable. Furthermore, remark that without the $c$-inner edges, we obtain a graph belonging to level 2 of the pushdown hierarchy.
\end{exemple}

\begin{figure}[h]
    \begin{minipage}[t]{.4\linewidth}
        \centering
        %\scalebox{0.65}{\input{grillediag.pstex_t}}
        \scalebox{0.65}{\includegraphics{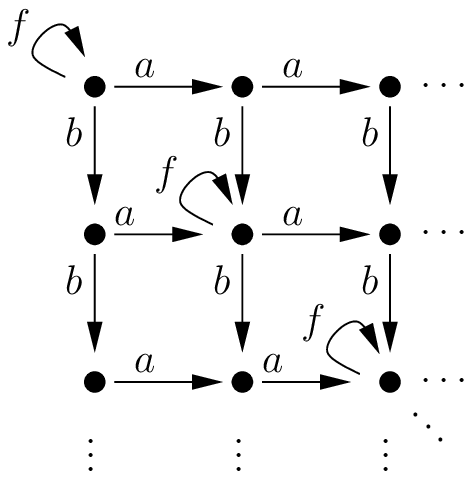}}
        \caption{The diagonal of the infinite grid (Example \ref{alexbis})}
        \label{fig:grillediag}
    \end{minipage}
    \hfill%
    \begin{minipage}[t]{.4\linewidth}
        \centering
        %\scalebox{0.32}{\input{cone.pstex_t}}
        \scalebox{0.32}{\includegraphics{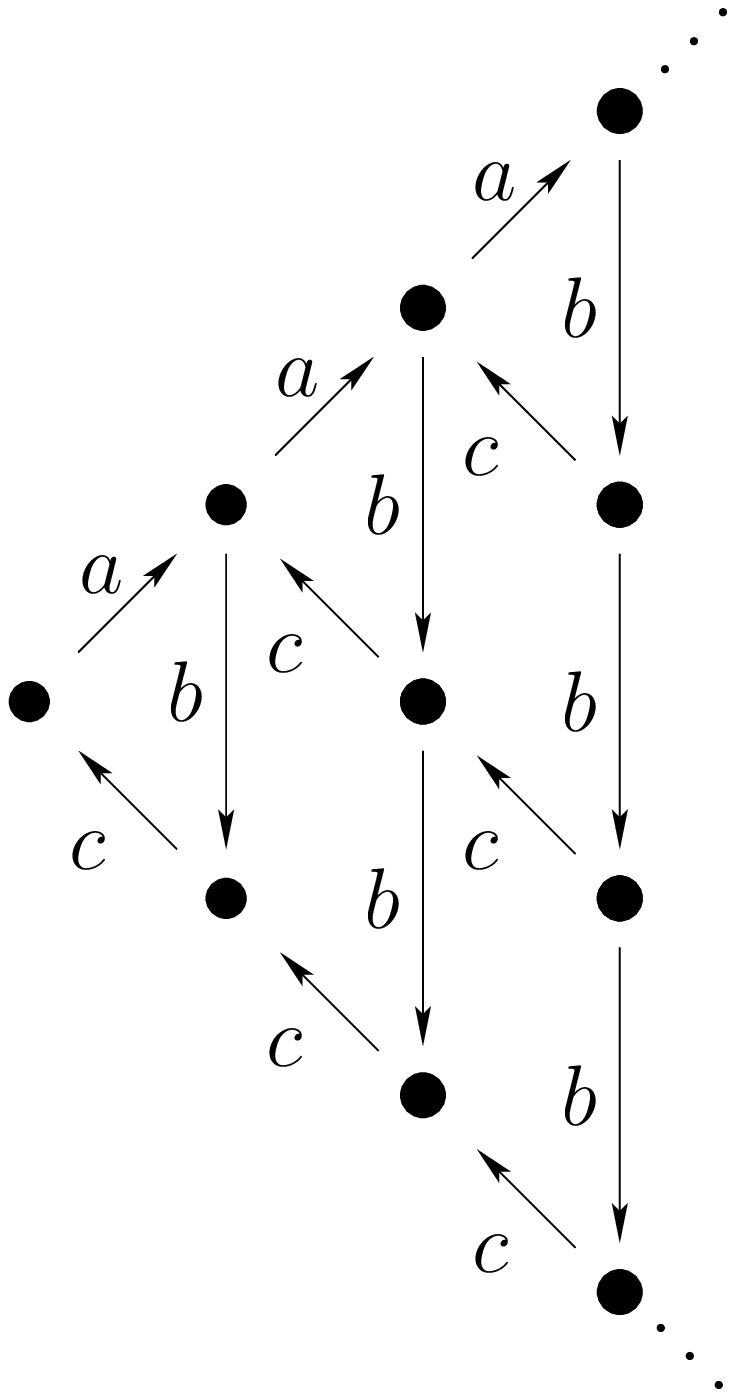}}
        \caption{The rewriting graph of a RTL (Example \ref{alex})}
        \label{fig:cone}
    \end{minipage}
\end{figure}

Before stating the main result (Theorem \ref{FOConfig}) of this section, we recall some basic definitions about word-automatic graphs.

  \paragraph{Word-automatic graphs.}Let $\Sigma$ be an alphabet and $\sharp\notin \Sigma$ a new symbol. The synchronization of two $\Sigma$-words, $u=a_1\dots a_m$ and $v=b_1\dots b_n$, is the $(\Sigma\ \dot \cup\ \{\sharp\})^2$-word $u\otimes v$ defined by \\ $u\otimes v:=(a_1,b_1)\dots (a_k,b_k)(x_{k+1},y_{k+1})\dots(x_K,y_K)$, where $k=\minimum(m,n)$, $K=\maximum(m,n)$ and for every $k< i \leqslant K$, $(x_i,y_i)=(\sharp,b_i)$ if $k=m$ and  $(x_i,y_i)=(a_i,\sharp)$ if not.

A $\Lambda$-graph ($\Lambda$ a finite alphabet) $G$ is \emph{word-automatic} if there exists a regular word language $L_{V_G}$ and a bijection  $\nu:L_{V_G}\longrightarrow V_G$ such that for each $\lambda\in \Lambda$, the synchronized word language $L_{\lambda}=\{\nu^{-1}(s)\otimes \nu^{-1}(t)\mid s\xrightarrow[G]{\lambda} t\}$ is regular.

 The following proposition recalls that the domain and the image of any word-automatic relation is regular.
 \begin{lemme}\label{projection}If a  language $L$ of $(\Sigma\ \dot \cup\ \{\sharp\})^2$-words $u\otimes v$ is regular, then the languages $\{u\in \Sigma^*\mid \exists v\in \Sigma^*\  u\otimes v\in L\}$ and $\{v\in \Sigma^*\mid \exists u\in \Sigma^*\ u\otimes v\in L\}$ are regular.
   \end{lemme}

   \begin{remarque}By Lemma \ref{projection}, a $\Lambda$-graph is word-automatic if and only if there exists a bijection $\nu:L \longrightarrow V$, where $L\in \Reg(\Sigma^*)$ and $V\supseteq V_G$ such that  for each $\lambda\in\Lambda$, the $(\Sigma\ \dot \cup\ \{\sharp\})^2$-word language $L_{\lambda}=\{\nu^{-1}(s)\otimes \nu^{-1}(t)\mid s\xrightarrow[G]{\lambda} t\}$ is regular.

     \end{remarque}

 \begin{remarque}\label{RegulierEtAutomatique}Let $L$ be a regular $P$-word language. Then the  $(P\ \dot \cup\ \{\sharp\})^2$-word language $\{u\otimes v\mid u,v\in L\}$ is regular. In particular, the $(I_D^{-}\ \dot\cup\ \{\sharp\})^2$-word language $\{\foata{s}\otimes\foata{t}\mid s,t\in M(\Sigma,D)\}$ is regular.

 \end{remarque}

The following theorem is partially due to the unique encoding of any trace by its Foata normal form.

\begin{theo}\label{FOConfig}
The rewriting graph of a recognizable trace rewriting with level-regular contexts (RTL graph) is word-automatic.
\end{theo}

Theorem \ref{FOConfig} is no more guaranteed if we suppose that left and right hand sides are just level-regular (see Remark \ref{cotereg}).

\begin{corollaire}[\cite{DBLP:journals/tcs/Hodgson82}]
  The $\FO$ theory of a RTL graph is decidable.
\end{corollaire}

In order to prove Theorem \ref{FOConfig}, we set out a crucial property about compatibility between concatenation and Foata normal forms.

In general, $\foata{st}\neq\foata{s}\foata{t}$. Indeed, suppose $D=\{(a,a),(b,b)\}$. If $s=[a]$ and $t=[ab]$, then $\foata{s}=\{a\}$, $\foata{t}=\{a,b\}$ and $\foata{st}=\{a,b\}\{a\}$. The following lemma expresses some compatibility between concatenation and Foata normal form.
\begin{lemme}\label{ConcatEtFoata}
Let $s,t\in M(\Sigma,D)$ such that $\foata{s}=A_1\cdots A_p$ ($p\geqslant 0$) and $\foata{st}=B_1\cdots B_m$. Then  $m\geqslant p$, $A_i\subseteq B_i$ for each $1\leqslant i \leqslant p$ and $\Pi_{I_D}((B_1\setminus A_1)\cdots (B_p\setminus A_p) B_{p+1}\cdots B_m)=t$.
\end{lemme}

\begin{proof} By induction on the length of $t$.
\end{proof}

 In the following, for $\foata{s}=A_1\cdots A_p$ ($p\geqslant 0$) and $t\in M(\Sigma,D)$, denote by $\foata{s}\parallel t$ the $I_D$-word language $B_1\cdots B_m$ ($m\geqslant p$) such that  $A_i\subseteq B_i$ for each $1\leqslant i \leqslant p$ and $\Pi_{I_D}((B_1\setminus A_1)\cdots (B_p\setminus A_p) B_{p+1}\cdots B_m)=t$. Thus  $\foata{st}\in \foata{s}\parallel t$, by the lemma above.

 \begin{exemple}
   Suppose $D=\{(a,a),(b,b)\}$ and consider $s=[aba]$ and $t=[ab]$. Then $\foata{s}=\{a,b\}\{a\}$ and $$\foata{s}\parallel t=\{a,b\}\{a\}\varnothing^*(\{a\}\varnothing^*\{b\}+\{b\}\varnothing^*\{a\}+\{a,b\})\varnothing^*\ \cup\ \{a,b\}\{a,b\}\varnothing^*\{a\}\varnothing^*$$
 \end{exemple}

\begin{proof}[Proof of Theorem \ref{FOConfig}]By Proposition \ref{foata} and Lemma \ref{levelautomata}, $\Pi_{\Foata}$ is a bijection from the regular language $\Foata$ onto $M(\Sigma,D)\supseteq V_{\graph_\Rrond}$. We are going to prove that for each $\lambda\in\Lambda$, the $(I_D\ \dot\cup\ \{\sharp\})^2$-word language  $L_{\lambda}=\{\foata{[u][v]}\otimes \foata{[u][w]}\mid [u]\in \rond{U},[v]\in \rond{V},[w]\in \rond{W},\ \rond{U}\cdot(\rond{V}\xrightarrow{\lambda} \rond{W})\in \Rrond\}$ is regular.
  \\Let $\rond{U}\cdot(\rond{V}\xrightarrow{\lambda} \rond{W})$ be a rule in $\Rrond$.  We have to prove that the $(I_D\ \dot\cup\ \{\sharp\})^2$-word language  $\{\foata{[u][v]}\otimes \foata{[u][w]}\mid [u]\in \rond{U},[v]\in \rond{V},[w]\in \rond{W}\}$ is regular. By Lemma \ref{ConcatEtFoata} and Remark \ref{RegulierEtAutomatique} and because the intersection of two regular word languages is regular, it suffices to show that the language of $(I_D\ \dot\cup\ \{\sharp\})^2$-words of the form $X\otimes Y$ such that there exists $[u]\in \rond{U}$, $[v]\in \rond{V}$ and $[w]\in \rond{W}$ such that $X\in \foata{[u]}\parallel [v]$ and $Y\in \foata{[u]}\parallel [w]$, is regular. For this, consider the  $I_D$-automata $\rond{A}_1$, $\rond{A}_2$ et $\rond{A}_3$ that recognize respectively $\{\foata{u}\mid [u]\in \rond{U}\}$, $\Pi_{I_D}^{-1}(\rond{V})$ and $\Pi_{I_D}^{-1}(\rond{W})$ and define the following $(I_D\ \dot\cup\ \{\sharp\})^2$-automaton.
  \begin{itemize}

    \item The initial state is $(i_{\rond{A}_1},i_{\rond{A}_2},i_{\rond{A}_3})$

    \item  the $(I_D\ \dot\cup\ \{\sharp\})^2$-graph is given by

      $(p,q,r)\xrightarrow{A\dot\cup B / A\dot\cup C} (p',q',r')$ : $p\xrightarrow[\rond{A}_1]{A} p'$, $q\xrightarrow[\rond{A}_2]{B} q'$, $r\xrightarrow[\rond{A}_3]{C} r'$

  $(p,q,r)\xrightarrow{B/C} (\bot,q',r')$ : $p\in F_{\rond{A}_1}\cup \{\bot\}$, $q\xrightarrow[\rond{A}_2]{B} q'$, $r\xrightarrow[\rond{A}_3]{C} r'$

    %$ (\bot,q,r)\xrightarrow{B/C} (\bot,q',r')$ : $q\xrightarrow[\rond{A}_2]{B} q'$, $r\xrightarrow[\rond{A}_3]{C} r'$

  $(p,q,r)\xrightarrow{\sharp/C} (\bot,\bot,r')$ : $p\in F_{\rond{A}_1}\cup \{\bot\}$, $q\in F_{\rond{A}_2}$, $r\xrightarrow[\rond{A}_3]{C} r'$

   %$(\bot,q,r)\xrightarrow{ \sharp/  C} (\bot,\bot,r')$ : $q\in F_{\rond{A}_2}$, $r\xrightarrow[\rond{A}_3]{C} r'$

  $(\bot,\bot,r)\xrightarrow{\sharp/C} (\bot,\bot,r')$ : $r\xrightarrow[\rond{A}_3]{C} r'$

    $(p,q,r)\xrightarrow{B/\sharp} (\bot,q',\bot)$ : $p\in F_{\rond{A}_1}\cup \{\bot\}$, $r\in F_{\rond{A}_3}$, $q\xrightarrow[\rond{A}_2]{B} q'$

    %$(\bot,q,r)\xrightarrow{B/ \sharp} (\bot,q',\bot)$ : $r\in F_{\rond{A}_3}$, $q\xrightarrow[\rond{A}_2]{B} q'$

    $(\bot,q,\bot)\xrightarrow{B/\sharp} (\bot,q'\bot)$ : $q\xrightarrow[\rond{A}_2]{B} q'$

  \item the set of  final states is $F=\{(p,q,r)\mid p\in F_{\rond{A}_1}\cup\{\bot\},q\in F_{\rond{A}_2},r\in F_{\rond{A}_3}\}\cup\{(\bot,\bot,r)\mid r\in F_{\rond{A}_3}\}\cup\{(\bot,q,\bot)\mid q\in F_{\rond{A}_2}\}$.

    \end{itemize}

  \end{proof}

\begin{remarque}\label{cotereg}Suppose $D=\{(a,a),(b,b),(c,c)\}$ and consider the following rewriting rule: $[(ab)^*]([\varepsilon]\xrightarrow{}[(bc)^*])$. Observe that $[(ab)^*]$ and $[(bc)^*]$ are level-regular but not recognizable. Recall that if a relation is word-automatic, then its image is regular (Proposition \ref{projection}). The rewriting graph of this rewriting rule fails to be word-automatic by encoding its vertex set by their Foata normal forms because $\Pi_{\Foata}^{-1}([(ab)^*]\cdot [(bc)^*])$ is not regular.

\end{remarque}

The $\FOAccs$ theory of a RTL graph may fail to be decidable. Indeed, the halting problem of 2-counter Minsky machines can be encoded by RTL graphs. 

\begin{prop}\label{Minsky}There exists some RTL graphs that does not have a decidable $\FOAccs$ theory.
\end{prop}

Before proving the proposition above, let us recall some basic definitions about 2-counter Minsky machines.

A 2-counter Minsky machine $M$ of length $n$ is a  sequence of $n$ instructions. The $n$-th instruction is a special instruction that halts the machine and for each $k\in\{1,\dots,n-1\}$ the $k$-th instruction is of the form

\begin{tabular}{|lll}
  $k:$ & $c:=c+1; \goto(j)$& ($\Incr(c,j)$)\\
\end{tabular}

or

\begin{tabular}{|lll}
  $k:$ & $\text{if } c\neq 0 \text{ then } c:=c-1; \goto(j)  \text{ else } \goto(l)$&($\Decr(c,j,l)$)
\end{tabular}
\\where $j,l\in\{1,\dots,n\}$ and $c$ is one of the 2 counters.

Configurations of $M$ are the triples $(k,c_1,c_2)\in\{1,\dots,n\}\times \N\times \N$, where $k$ is the instruction number, and $c_1$ and $c_2$ the 2-counter contents. The initial configuration is $(1,0,0)$. A computation is a sequence of configurations starting from the initial configuration and such that two successive configurations respect the instructions. The halting problem is: given a 2-counter Minsky machine, is there a finite computation that halts the machine ?

\begin{theo}[Minsky]
The halting problem of 2-counter Minsky machines is undecidable.
\end{theo}
              
\begin{proof}[Proof of Proposition \ref{Minsky}]

Given a 2-counter Minsky machine  $M$ of length $n$, consider the rewriting graph $G_M$ of the following recognizable trace rewriting system:
  
\begin{itemize}
\item $\Sigma:=\{\bot_a,\bot_b,a,b,1,\dots,n\}$

\item the independence relation $I$ on $\Sigma$ is given by: $aIb$, $\bot_aI\bot_b$

\item for each $k\in\{1,\dots,n-1\}$ the rewriting rules are:
  \begin{itemize}
  \item $[k]\xrightarrow{R} [cj]$ ($j\in\{1,\dots,n\}$, $c\in\{a,b\}$) if the $k$-th instruction is $\Incr(c,j)$
  \item $[ck] \xrightarrow{R} [j]$ 
  
  and 
  
  $[\bot_ck]\xrightarrow{R} [\bot_cl]$ ($j,l\in\{1,\dots,n\}$, $c\in\{a,b\}$) if the $k$-th instruction is $\Decr(c,j,l)$

 \item $[\bot_a\bot_b1]([\varepsilon]\xrightarrow{i}[\varepsilon])$

\item $[n]([\varepsilon]\xrightarrow{f}[\varepsilon])$

  \end{itemize}

\end{itemize}

The initial configuration is encoded by $[\bot_a\bot_b1]$. Final configurations are encoded by $[\bot_a\bot_ba^*b^*n]$.
A configuration $(k,c_1,c_2)$ accessible from $[\bot_a\bot_b1]$ is encoded by the trace $[\bot_a\bot_b\overbrace{a\dots a}^{c_1}\overbrace{b\dots b}^{c_2}k]$.
\\The machine $M$ halts if and only if $G_M$ satisfies: $\exists x\exists y (x\xrightarrow{i}x\ \land\  y\xrightarrow{f}y\ \land\ x\xrightarrow{*}y)$.

\end{proof}

\section{Concurrent unfolding of a concurrent automaton}In this section, we consider concurrent automata, that were first introduced in \cite{Shields1997} as asynchronous transition systems, and we prove that the $\FOAccs$ theory of their concurrent unfoldings is decidable. Indeed, we will show that the concurrent unfolding of a concurrent automaton, with the reachability relation is a RTL graph.

 Let $(\Sigma,D)$ be a dependence alphabet and $I=\Sigma^2\setminus D$.
 
\subsection{Concurrent automata}
\begin{df}An $\Sigma$-automaton $\rond{A}=(G,i,F)$ is $D$-concurrent when
   \begin{itemize}
   \item $G$ is deterministic 
   \item $((a,b)\in I \text{ and } p\xrightarrow[\rond{A}]{ab}q) \ \Longrightarrow \ p\xrightarrow[\rond{A}]{ba}q$.
   \end{itemize}
 \end{df}

Every automaton can  be seen as a concurrent automaton relatively to the total dependence relation on its edge label set.

 \begin{exemple}\label{automateresidu}Let $\rond{L}\subseteq M(\Sigma,D)$ be a trace language. The residual automaton of $\rond{L}$ by $\Sigma$ is the $D$-concurrent $\Sigma$-automaton $\residual{\rond{L}}{\Sigma}$ defined by:
    \begin{itemize}
    \item the $\Sigma$-graph $\{[u]^{-1}\rond{L}\xrightarrow{a} [ua]^{-1}\rond{L}\mid u\in \Sigma^*, a\in \Sigma\}$
    \item the initial state $\rond{L}$

    \item final states $[u]^{-1}\rond{L}$ such that $[\varepsilon]\in [u]^{-1}\rond{L}$,
    \end{itemize}

is a $D$-concurrent $\Sigma$-automaton that recognises $\cup\rond{L}$ (see Figure \ref{fig:automateresidus}).  
   
  \end{exemple}

\begin{exemple}Let $\rond{L}\subseteq M(\Sigma,D)$ be a trace language. The unfolding automaton $U(\rond{L},\Sigma)$ of $\rond{L}$ by $\Sigma$ defined by

  \begin{itemize}
  \item the Cayley graph of $M(\Sigma,D)$: $\{[u]\xrightarrow{a}[ua]\mid u\in\Sigma^*,a\in \Sigma\}$
  \item the initial state $[\varepsilon]$
    \item final states $t\in \rond{L}$
    \end{itemize}
is a $D$-concurrent $\Sigma$-automaton that recognises $\cup\rond{L}$.
  \end{exemple}

 \begin{figure}
   \centering
    %\scalebox{0.7}{\input{residualSigmaBis.pstex_t}}
    \scalebox{0.7}{\includegraphics{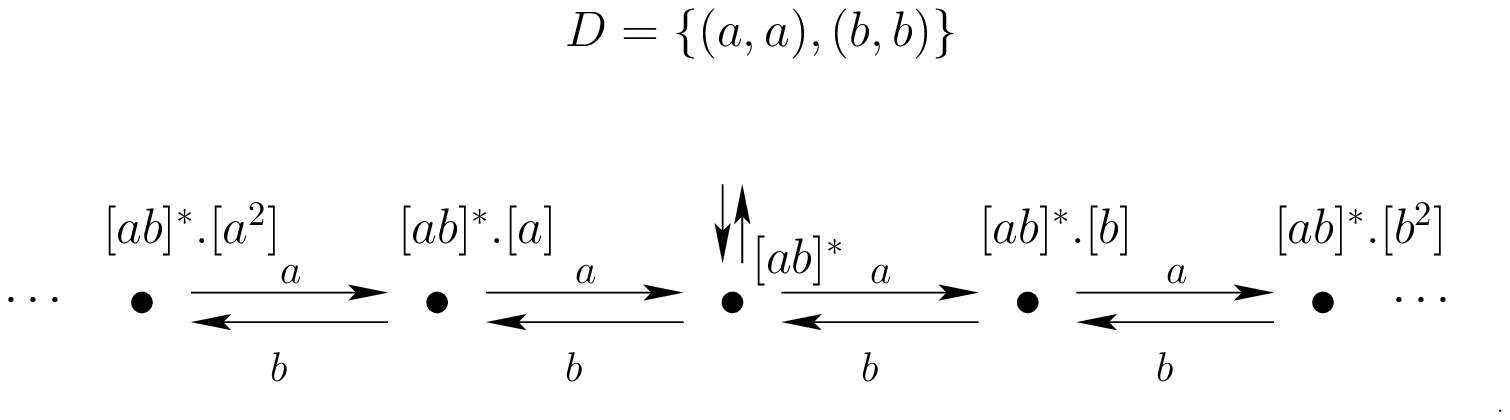}}
    \caption{$\residual{[(ab)^*]}{\Sigma}$ (Example \ref{automateresidu})}
   \label{fig:automateresidus}
  \end{figure}

By combining Proposition \ref{Recreg} and Example \ref{automateresidu}, we obtain the following characterization of recognizable trace languages: 
  
\begin{prop}
  A trace language $\rond{L}$ is recognizable if and only if there exists a finite $D$-concurrent $\Sigma$-automaton $\rond{A}$ such that $\cup \rond{L}=L(\rond{A})$.
\end{prop}

\subsection{The concurrent unfolding of a concurrent automaton}

\begin{df}
The $D$-unfolding $\U_D(\rond{A})$ of a $D$-concurrent $\Sigma$-automaton $\rond{A}$ is the $D$-concurrent $\Sigma\ \dot\cup\ \{f\}$-graph defined by: $$\U_D(\rond{A})=\{[u]\xrightarrow{a} [ua]\mid u\in \Sigma^*,a\in\Sigma,i\xrightarrow[\rond{A}]{ua}\}\ \cup\ \{[u]\xrightarrow{f} [u]\mid u\in L(\rond{A})\}.$$
\end{df}

In the following example, we introduce the infinite grid tree as the concurrent unfolding of a finite concurrent automaton.

\begin{exemple}
Let $\Sigma=\{a,b,c\}$ and suppose $aIb$. Consider the graph $G=\{p\xrightarrow{a,b,c}p\}$. The $D$-unfolding of $\rond{A}=(G,p,\varnothing)$ (Figure \ref{infinitegridtree}), is the infinite grid (on $\{a,b\}$) tree (see Section \ref{graphtree}) and has a decidable $\FOAccs$ theory by the theorem below.
\end{exemple}

Before stating the main result of this section, recall that the unfolding of a finite graph is a regular tree whose monadic second-order theory is decidable (since unfolding preserves monadic second-order decidability). Here, we consider a notion of concurrent unfolding and we apply this graph transformation to a wider class than the class of finite graphs.
\begin{theo}\label{FOReach}
If $\rond{A}$ is a finite $D$-concurrent automaton, then the  $\FOAccs$ theory of $\U_D(\rond{A})$ is decidable.
\end{theo}

We do not know if, in general, the $D$-unfolding preserves $\FOAccs$ decidability.

\begin{proof}[Proof of Theorem \ref{FOReach}]

Consider  the $\Sigma\ \dot \cup\ \{*\}$-automaton
$$\U_D(\rond{A})_*:=\U_D(\rond{A})\cup\{[u]\xrightarrow{*} [uv]\mid u,v\in \Sigma^*,i\xrightarrow[\rond{A}]{uv}\}$$
It is the rewriting graph of the following recognizable trace rewriting system:

\smallskip

   $\left\{
     \begin{array}{l}
          [L(G,i,Q_a)]([\varepsilon]\xrightarrow{a} [a]) \ \  a\in \Sigma\ \text{ and }\ Q_a=\{q\in Q\mid q\xrightarrow[\rond{A}]{a}\}\\
   
   [L(G,i,F)]([\varepsilon]\xrightarrow{f} [\varepsilon])\\
   
       [L(G,i,q)]([\varepsilon]\xrightarrow{*} [L(G,q,Q)]) \ \ q\in Q.\\
       
     \end{array} \right.$

\end{proof}

\begin{figure}
\centering
   %\scalebox{0.7}{\input{arbregrilleUFC.pstex_t}}
\scalebox{0.7}{\includegraphics{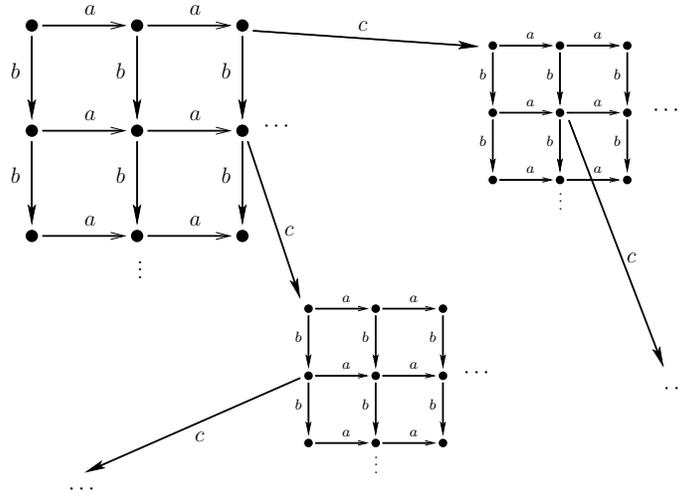}}
\caption{The infinite grid tree}
\label{infinitegridtree}
\end{figure}

\begin{remarque}Given a $\Sigma$-graph $G$, the $\FO$ theory of the graph $G\cup\{p\xrightarrow{\rond{L}}q\mid p\xrightarrow{u}q, [u]\in \rond{L},\rond{L}\in \REC(M(\Sigma,D))\}$ is refered to as the first-order theory with recognizable reachability predicates ($\FO[\REC]$) of $G$. We can strengthen the last theorem and show that: if $\rond{A}$ is a finite $D$-concurrent automaton, then $\U_D(\rond{A})$ has a decidable $\FO[\REC]$ theory. Indeed, observe that each sentence in $\FO[\REC]$ logic contains  a finite number of atomic formula $x\xrightarrow{\rond{L}_1}y$,\dots, $x\xrightarrow{\rond{L}_n}y$ ($n\geqslant 1$). Then $\U_D(\rond{A})\cup\{p\xrightarrow{\rond{L}_j}q\mid p\xrightarrow{u}q, [u]\in \rond{L}_j,j\in\{1,\dots ,n\}\}$ is the rewriting graph of the following RTL:
  
\begin{center}
   $\left\{
     \begin{array}{l}
          [L(G,i,Q_a)]([\varepsilon]\xrightarrow{a} [a]) \ \  a\in \Sigma\ \text{ and }\ Q_a=\{q\in Q\mid q\xrightarrow[\rond{A}]{a}\}\\
   
   [L(G,i,F)]([\varepsilon]\xrightarrow{f} [\varepsilon])\\
   
       [L(G,i,q)]([\varepsilon]\xrightarrow{\rond{L}_j} [L(G,q,Q)]\cap \rond{L}_j) \ \ q\in Q \ \ j\in \{1,\dots,n\}.\\
       
     \end{array} \right.$
 \end{center}
 \end{remarque}
 
We have deduced the $\FO[\REC]$ theory decidability of the Cayley graph of a trace monoid from the $\FO$ decidability of RTL graphs. The following remark shows the inverse reduction.
 \begin{remarque}\label{forec}
Lastly, note that any rewriting graph of a recognizable trace rewriting system (with recognizable contexts) on some trace monoid $M(\Sigma,D)$ is a $\FO[\REC]$ interpretation of the Cayley graph of this trace monoid. Indeed, observe that the neutral element is $\FO$-definable: $\neutral(x)=\forall t \bigwedge_{a\in\Sigma}\lnot(t\xrightarrow{a}x)$. Then for each rule of the form $\rond{U}\cdot(\rond{V}\xrightarrow{}\rond{W})$ consider the formula: $\phi(x,y)=\exists i\exists z(\neutral(i)\ \land\ i\xrightarrow{\rond{U}}z \land z\xrightarrow{\rond{V}}x \land z\xrightarrow{\rond{W}}y)$.
\end{remarque}

\subsection{Regular trace event structure} In \cite{DBLP:conf/lics/Madhusudan03}, Madhusudan proves that the $\FO$ theory of a regular trace event structure is decidable. For this, he shows that the vertex set and the relations of such a graph can be encoded by a recognizable trace language on a judicious dependence alphabet. Note that, due to the level-regular contexts, this technique does not allow to prove that the $\FO$ theory of RTL graphs is decidable.

A trace $t=[a_1\cdots a_n]\in M(\Sigma,D)$ is \emph{prime} if the set $\{1,\dots,n\}$, partially ordered by the relation $E$ defined by $iEj$ if and only if $i<j$ and $a_iDa_j$, has exactly one maximal element.
 
 Let $\rond{L}\subseteq M(\Sigma,D)$ be a trace language. Denote by $\premiere(\rond{L})$ the set of prime traces in $\rond{L}$.

\begin{df}
  The event structure defined by $\rond{L}$, $\StructEv_{\rond{L}}$, is the  $\{\leqslant,\sharp,(\lambda_a)_{a\in\Sigma} \}$-graph whose vertex set is $\premiere(\rond{L})$ defined  by 
\begin{itemize}
\item $t\xrightarrow{\leqslant} t'$ : $t\sqsubseteq t'$
\item $t\xrightarrow{\sharp} t'$ : $\forall t''\in \premiere(\rond{L})(t\not\sqsubseteq t'' \lor t'\not\sqsubseteq t''))\}$ 
\item $t\xrightarrow{\lambda_{a}} t$ : $\text{the maximal element of }t \text{ is }a$.
\end{itemize}
\end{df}

 \begin{theo}[\cite{DBLP:conf/lics/Madhusudan03}]If $\rond{L}\in \REC(M(\Sigma,D))$, then $\StructEv_{\rond{L}}$ has a decidable $\FO$ theory.
  \end{theo}
  
\begin{proof}
    A trace $t\in V_{\U_D(\residual{\rond{L}}{\Sigma})_*}$ is prime if and only if $t$ is not successor of two distinct vertices of $\U_D(\residual{\rond{L}}{\Sigma})_*$. Since this last property is $\FO$ expressible, the event structure $\StructEv_{\rond{L}}$ can be obtained by a  $\FO$ interpretation of $\U_D(\residual{\rond{L}}{\Sigma})_*$, that has a decidable $\FO$ theory.
\end{proof}

\section{Graph tree}\label{graphtree}
In this section, we consider ground term rewriting graphs. These graphs have a decidable $\FOAccs$ theory \cite{DBLP:conf/lics/DauchetT90}. We define a notion of graph tree and we prove that if a graph tree is a ground term rewriting graph (GTR graph), then it is finitely decomposable by size. A direct consequence is that the infinite grid tree, defined above as the concurrent unfolding of a concurrent automaton (Exemple \ref{infinitegridtree}), is not a GTR graph, although it has a $\FOAccs$ theory decidable.

\subsection{Ground Term Rewriting graphs (GTR graphs)}
A \emph{position} is an element of $\N^*$, the set of finite words over $\N$. Denote by $\sqsubseteq$ the prefix ordering over $\N^*$. Let $F$ be a ranked alphabet (each symbol in $F$ has an arity in $\N$). A \emph{term $t$ on $F$} is a partial function $t:\N^*\longrightarrow{F}$ whose domain, $\pos(t)$, has the following properties: \begin{itemize}
\item $\pos(t)\neq\varnothing$
  \item $\pos(t)$ is prefix closed  \item $\forall u\in \pos(t)$, if the arity of $t(u)$ is $n$ ($n\geqslant 0$), then $\{j\mid uj\in\pos(t)\}=\{1,\dots,n\}$.
\end{itemize}
The \emph{size} $\taille{t}$ of a term $t$ is the number of its nodes. The \emph{subterm} of $t$ at position $u$, denoted $\ssterme{t}{u}$, is the term on $F$ defined by: \begin{itemize}
  \item $\pos(\ssterme{t}{u})=\{v\in \N^*\mid uv\in\pos(t)\}$
  \item $\forall v\in \pos(\ssterme{t}{u})$, $(\ssterme{t}{u})(v)=t(uv)$.
\end{itemize}
If $u\in\pos(t)$ and $s$ is a term, then \emph{$t\sub{u}{s}$}, the term obtained from $t$ by replacing  the subterm $\ssterme{t}{u}$ by $s$, is defined by : $$t\sub{u}{s}(v)=\left\{\begin{aligned}
&s(w) \text{ if } v=uw \text{ and } w\in \pos(s)\\
&t(v) \text{ if }v\in\pos(t) \text{ and }u\not \sqsubseteq v\end{aligned}\right.$$
If $t$ is a term on $F$ and $u\in\pos(t)$, then the \emph{context of $t$ at the position $u$} is the term $t\sub{u}{x}$ on $F\ \dot\cup\ \{x\}$, where $x$ is a constant \textit{i.e} the arity of $x$ is 0.

A \emph{context} $C$ on $F$ is a term on $F\ \dot\cup\ \{x\}$, $x$ constant, such that there exists a unique position $u_C\in\pos(C)$ for which $C(u_C)=x$. If $t$ is a term on $F$, then the term $C[t]$ on $F$ is defined by $C[t]:=C\sub{u_C}{t}$. The size $\taille{C}$ of a context $C$ on $F$ is the number of its nodes minus 1.

A \emph{ground term rewriting system} $\GTrs{R}$ is a $4$-tuple $\GTrs{R}=(F,\Sigma,R,i)$ where:\begin{itemize}
  \item $F$ is a ranked alphabet  \item $\Sigma$ is a label alphabet
  \item $R:={\displaystyle \bigcup_{a\in\Sigma}} R_a$, where for each $a\in \Sigma$, $R_a$ is a finite set of rules of the form $s\xrightarrow{a}s'$ with $s$ and $s'$ distinct terms on $F$
  \item $i$ is an initial $F$-term.
\end{itemize}

We write: \begin{itemize}
\item $t\xrightarrow[\GTrs{R}]{a}t'$ if  there exists a position $p\in \pos(t)$ and a rule $s\xrightarrow{a}s' \in R_a$ such that $\ssterme{t}{p}=s$ and $t'=t\sub{p}{s'}$
\item $t\xrightarrow[\GTrs{R}]{}t'$ when there exists $a\in \Sigma$ such that $t\xrightarrow[\GTrs{R}]{a}t'$
\item $\xrightarrow[\GTrs{R}]{*}$ for the reflexive and transitive closure under composition of $\xrightarrow[\GTrs{R}]{}$.
\end{itemize}
The \emph{configuration graph} $\graph_{\GTrs{R}}$ of $\GTrs{R}$ is the $\Sigma$-graph defined by $$\graph_{\GTrs{R}}:=\{t\xrightarrow[\GTrs{R}]{a}t'\mid i\xrightarrow[\GTrs{R}]{*}t,a\in \Sigma\}$$
A graph is called Ground Term Rewriting graph (GTR graph) if it is isomorphic to the configuration graph of a ground term rewriting system. 
\begin{remarque}\label{loop}GTR graphs have no loop since each rule in the rewriting system has distinct left hand side and right hand side.
\end{remarque}

\begin{exemple}
The infinite grid is a GTR graph (see Figure \ref{grilleinfinieGTRS}).
\end{exemple}

\begin{figure}
\centering
   %\scalebox{0.7}{\input{grilleGTRSUFC.pstex_t}}
   \scalebox{0.7}{\includegraphics{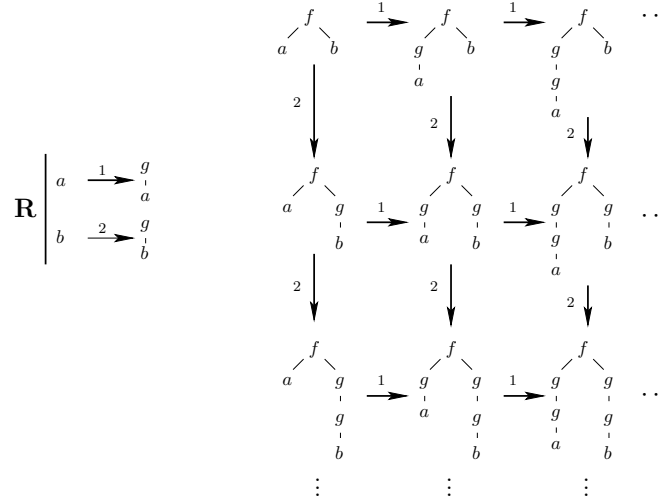}}
\caption{The infinite grid is a GTR graph}
\label{grilleinfinieGTRS}
\end{figure}

In \cite{DBLP:conf/lics/DauchetT90}, Dauchet and Tison prove that a GTR graph with the reachability relation is tree-automatic. Thus:

\begin{theo}[ \cite{DBLP:conf/lics/DauchetT90}]
GTR graphs have a decidable $\FOAccs$ theory.
\end{theo}

\subsection{Finite decomposition of a graph}
Let us start by recalling the definition of the frontier of a subgraph.
\begin{df}\label{frontiere}
Let $G$ be a graph and $H\subseteq G$ a subgraph of $G$. The frontier of $H$ (in $G$) is $\Fr(H)=V_H \cap V_{G-H}$.
\end{df}

The frontier of $H$ is the set of $H$-vertices that are incident to an edge in $G-H$.

Let $\graph_{\GTrs{R}}$ be a GTR graph. For each $n\geqslant0$, $$\GRrondn:=\{s\xrightarrow{e}t\in \graph_{\GTrs{R}}\mid \taille{s}<n \text{ or } \taille{t}<n\}$$

According to Definition \ref{frontiere}, the frontier of $\graph_{\GTrs{R}}-\GRrondn$ is $\Fr({\graph_{\GTrs{R}}}-\GRrondn)=V_{{\graph_{\GTrs{R}}}-\GRrondn}\cap V_{\GRrondn}$. And the frontier of $K$, a connected component of ${\graph_{\GTrs{R}}}-\GRrondn$, is $\Fr(K)=\Fr({\graph_{\GTrs{R}}}-\GRrondn)\cap V_K$. The frontier of $K$ is formed by the $K$-vertices incident to an edge in $\GRrondn$.

\medskip
The graph $\graph_{\GTrs{R}}$ is \emph{finitely decomposable by size} if $$\dec:=\{(K,\Fr(K))\mid K \text{ connected component of } \graph_{\GTrs{R}}-G_n, n\geqslant 0\}$$ has finite index, for the isomorphism relation.

\begin{theo}[\cite{DBLP:conf/birthday/Caucal08}]\label{finitelydecomposable}
If a countable graph is finitely decomposable by size, then it is at the first level of the pushdown hierarchy. In particular, it has a decidable MSO theory.
\end{theo}

\subsection{Graph tree and finite decomposition}

\begin{df}Let $G$ be a $\Sigma$-graph and $p_0\in V_G$. Given a new symbol $c\notin \Sigma$, the $G$-tree from $p_0$ is the $\Sigma\ \dot\cup\ \{c\}$-graph, $\arbre{G}{p_0}$, defined by $$\arbre{G}{p_0}:=\{up\xrightarrow{a}uq\mid u\in V_G^*, p\xrightarrow[G]{a}q\}\cup\{u\xrightarrow{c}up_0\mid u\in V_G^*\}$$
\end{df}

\begin{exemple}
See Figure \ref{arbredemidroite} for the semi-line tree.
\end{exemple}

\begin{figure}
\centering
   %\scalebox{0.6}{\input{arbredemidroiteUFC.pstex_t}}
   \scalebox{0.6}{\includegraphics{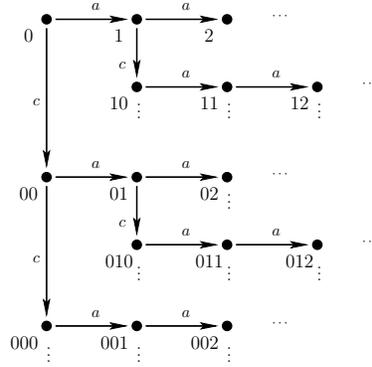}}
\caption{The semi-line tree}
\label{arbredemidroite}
\end{figure}

\begin{remarque}\label{cdeterminisme}
The graph $\arbre{G}{p_0}$ is $c$-deterministic: $$(v\xrightarrow[\arbre{G}{p_0}]{c}v_1\text{ and }v\xrightarrow[\arbre{G}{p_0}]{c}v_2)\Longrightarrow v_1=v_2$$\end{remarque}

\begin{remarque}
The graph $\arbre{G}{p_0}$ is a tree if and only if $G$ is a tree.\end{remarque}

\begin{remarque}Let $(\Sigma,D)$ be a dependence alphabet, $G$ a finite $D$-concurrent $\Sigma$-graph and $p_0\in V_G$. The $\Sigma\  \dot \cup\  \{c\}$-graph defined by $G\cup \{p\xrightarrow{c}{p_0}\mid p\in V_G\}$ is $D_c$-concurrent, with $D_c=D\ \cup\ (\Sigma \cup\{c\})\times(\Sigma \cup\{c\})$. Its $D_c$-unfolding from $p_0$ is $\arbre{\U_D(G,p_0)}{p_0}$. 

\end{remarque}

\begin{theo}\label{theoprincipal}
If $\arbre{G}{p_0}$ is a GTR graph, then $\arbre{G}{p_0}$ is finitely decomposable by size.\end{theo}

The MSO theory of the infinite grid is undecidable. The same holds for the infinite grid tree. By combining Theorem \ref{theoprincipal} and Theorem \ref{finitelydecomposable}, we deduce the corollary below.

\begin{corollaire}
The infinite grid tree is not a GTR graph.
\end{corollaire}

\subsection{Proof of Theorem \ref{theoprincipal}}
\begin{lemme}\label{reecincomp}Let $G$ be a $\Sigma$-graph and $p_0\in V_G$. If there exists a ground term rewriting system $\GTrs{R}=(F,\Sigma,R,i)$ such that  $\arbre{G}{p_0}$ is isomorphic to $\graph_{\GTrs{R}}$, then for every term $t\in V_{\graph_{\GTrs{R}}}$, there exists a smallest position $u_t$ (for the prefix ordering $\sqsubseteq$) at which $t$ is incident to a rewriting in $\graph_{\GTrs{R}}$.
\end{lemme}

\begin{proof}[Proof of Lemma \ref{reecincomp}]
It suffices to prove that if there exists two incomparable positions $u'$ and $u''$ at which $t$ is incident to rewritings, then there exists a position $v$, $v\sqsubseteq u'$, $v\sqsubseteq u''$ at which $t$  is incident to a rewriting.

Denote by $e'$ (respectively $e''$) the label of the rewriting $t$ is incident in position $u'$ (respectively $u''$). We are going to show that $c\notin \{e',e''\}$. Since $u'$ and $u''$ are incomparable and because of Remark \ref{loop}, there exists two paths between two distinct vertices of $\graph_{\GTrs{R}}$, labelled by $e'e''$ and $e''e'$, each of them with no loop (see Figure \ref{reecincompdessin}). This is possible in $\arbre{G}{p_0}$ only if $c\notin \{e',e''\}$. Indeed,
\begin{itemize}
\item $\{e',e''\}\subseteq \{c\}$ is impossible because of Remark \ref{cdeterminisme}
\item $e'=c$ and $e''\in\Sigma$ (or the converse $e''=c$ and $e'\in\Sigma$) is impossible because $c$ and $e''$ do not commute in $\arbre{G}{p_0}$.
\end{itemize} 

\begin{figure}
\centering
\psset{arrowsize=5pt}

\psset{xunit=0.95cm,yunit=0.95cm}
\begin{pspicture}(0,0)(5,4)
%\psgrid
\pscircle[linewidth=2pt](2,2){1.6}
\psline(1,1.2)(2,3.4)(3.2,1.2)(1,1.2)

\psline(2,3.4)(2,2.6)(1.6,1.6)
\psline(1.6,1.6)(1.4,1.2)
\psline(1.6,1.6)(2,1.2)

\psline(1.8,2)(2.4,1.8)

\psline(2.4,1.8)(2.2,1.2)
\psline(2.4,1.8)(2.6,1.2)
\rput(1.4,1.6){{\tiny $u'$}}

\rput(2.2,1.7){{\tiny $u''$}}

\psline[linestyle=dashed,linewidth=.4pt](2,2.6)(3.4,2.6)

\psline{->}(3.5,2.6)(5,2.6)
\rput(4.2,2.8){{$c$}}

\psline[linestyle=dashed,linewidth=.4pt](2.4,1.8)(3.6,1.8)

\psline(3.6,1.8)(4.8,1.8)
\rput(4.2,2){$e''$}

\psline[linestyle=dashed,linewidth=.4pt](1.6,1.6)(3.2,1)

\psline(3.2,1)(4.4,0.6)
\rput(4,1){$e'$}
\end{pspicture}

%\vspace{1cm}
\begin{pspicture}(0,-1)(3,3)
%\psgrid
\rput(1,1){\rnode{t}{$t$}}
\rput(0,0){\rnode{u}{$t'$}}
\rput(2,0){\rnode{v}{$t''$}}
\rput(1,-1){\rnode{w}{$\bullet$}}

\psset{nodesep=2pt}
{\scriptsize 
\ncline[arrows=->]{t}{u}\nbput{$e'$}\ncline[arrows=->]{t}{v}\naput{$e''$}

\ncline[arrows=->]{u}{w}\nbput{$e''$}\ncline[arrows=->]{v}{w}\naput{$e'$}
}
%\rput(1,2.4){\rnode{A}{$\xymatrix{t'\ar@{<-}[r]^{e'}_{u'}&t\ar@{<-}[r]^{e''}_{u''}&t''}$}}

\rput(1,2.4){\rnode{A}{$t'\xleftarrow[u']{e'}t\xrightarrow[u'']{e''}t''$}}

\ncline[nodesep=4pt,linewidth=2pt,arrows=->]{A}{t}
\end{pspicture}
\hfill\begin{pspicture}(0,-1)(3,3)
%\psgrid
\rput(1,1){\rnode{t}{$t$}}
\rput(0,0){\rnode{u}{$t'$}}
\rput(2,0){\rnode{v}{$t''$}}
\rput(1,-1){\rnode{w}{$\bullet$}}

\psset{nodesep=2pt}
{\scriptsize
\ncline[arrows=->]{u}{t}\naput{$e'$}\ncline[arrows=->]{v}{t}\nbput{$e''$}

\ncline[arrows=->]{w}{u}\naput{$e''$}\ncline[arrows=->]{w}{v}\nbput{$e'$}
}

%\rput(1,2.4){\rnode{A}{$\xymatrix{t'\ar@{->}[r]^{e'}_{u'}&t\ar@{<-}[r]^{e''}_{u''}&t''}$}}

\rput(1,2.4){\rnode{A}{$t'\xrightarrow[u']{e'}t\xleftarrow[u'']{e''}t''$}}

\ncline[nodesep=4pt,linewidth=2pt,arrows=->]{A}{t}
\end{pspicture}
\hfill\begin{pspicture}(0,-1)(3,3)
%\psgrid
\rput(1,1){\rnode{t}{$t$}}
\rput(0,0){\rnode{u}{$t'$}}
\rput(2,0){\rnode{v}{$t''$}}
\rput(1,-1){\rnode{w}{$\bullet$}}

\psset{nodesep=2pt}
{\scriptsize 
\ncline[arrows=->]{u}{t}\naput{$e'$}\ncline[arrows=->]{t}{v}\naput{$e''$}

\ncline[arrows=->]{u}{w}\nbput{$e''$}\ncline[arrows=->]{w}{v}\nbput{$e'$}
}
%\rput(1,2.4){\rnode{A}{$\xymatrix{t'\ar@{->}[r]^{e'}_{u'}&t\ar@{->}[r]^{e''}_{u''}&t''}$}}

\rput(1,2.4){\rnode{A}{$t'\xrightarrow[u']{e'}t\xrightarrow[u'']{e''}t''$}}

\ncline[nodesep=4pt,linewidth=2pt,arrows=->]{A}{t}
\end{pspicture}
\caption{Paths in $\graph_{\GTrs{R}}$ between two distinct vertices, labelled by $e'e''$ and $e''e'$}
\label{reecincompdessin}
\end{figure}
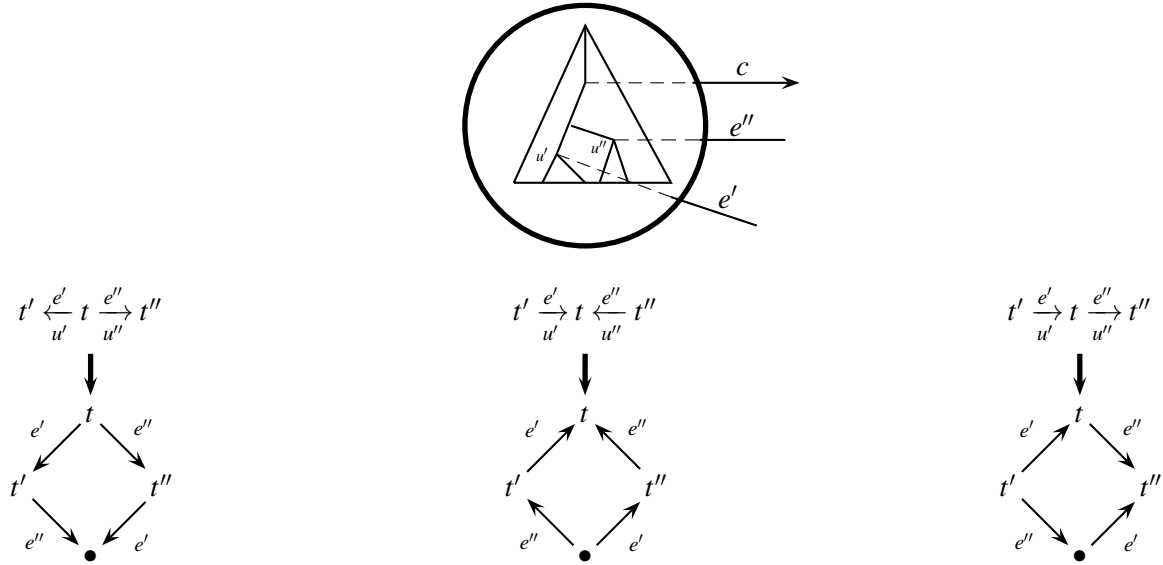

But there exists a position $v$ at which the term $t$ is incident to a rewriting labelled by $c$. Due to the precedent point, the position $v$ must be comparable to positions $u'$ and $u''$. Since $u'$ and $u''$ are not comparable, we deduce that $v\sqsubseteq u'$ and $v\sqsubseteq u''$.
\end{proof}

\begin{proof}[Proof of Theorem \ref{theoprincipal}]

Let $\GTrs{R}=(F,\Sigma,R,i)$ be a ground term rewriting system such that $\arbre{G}{p_0}$ is isomorphic to $\graph_{\GTrs{R}}$. We have to show that $$\dec:=\{(K,V_K\cap V_{\GRrondn})\mid K \text{ connected component of } \graph_{\GTrs{R}}-\GRrondn, n\geqslant 0\}$$ has finite index. Let $\delta:=\max\{\lvert\taille{d}-\taille{g}\rvert\mid g\xrightarrow{e}d\in R\}$ and $M:=\max\{\taille{g},\taille{d}\mid g\xrightarrow{e}d\in R\}$. We are going to show that for each connected component $K$ in $\dec$, there exists a position $u_K$ and a context $C_K$ such that  \begin{itemize}

\item for every term $t\in V_K$, $u_K\in \pos(t)$ and $C_K$ is the context of $t$ at the position $u_K$ ($t=C_K[\ssterme{t}{u_K}]$)

\item for every term $t\in  \Fr_{\graph_{\GTrs{R}}}(K)$, $\taille{\ssterme{t}{u_K}}<M+\delta$.
\end{itemize}Then the finite subset of the (finite) set of terms whose size is at most $M+\delta$, obtained from $\Fr_{\graph_{\GTrs{R}}}(K)$ by removing the context $C_K$, is 
characteristic of the isomorphy type of $(K,\Fr_{\graph_{\GTrs{R}}}(K))$. Indeed, 
for $K\in \dec$, let $\widetilde{K}:=\{s\mid C_K[s]\in \Fr_{\graph_{\GTrs{R}}}(K)\}$. If $\widetilde{K}=\widetilde{K'}$, then $(K,\Fr_{\graph_{\GTrs{R}}}(K))$ and $(K',\Fr_{\graph_{\GTrs{R}}}(K'))$ are isomorphic \textit{via} $C_K[s]\mapsto C_{K'}[s]$.

\medskip
Let $K\in \dec$ and $n\geqslant 0$ such that $K$ is a connected component of $\graph_{\GTrs{R}} - \GRrondn$. Remark that for every $t\in \Fr_{\graph_{\GTrs{R}}}(K)$, $n\leqslant \taille{t}< n+\delta$. In particular, $\Fr_{\graph_{\GTrs{R}}}(K)$ is finite.
\\Let $m_K:=\min\{\taille{t}\mid t\in V_K\}$. Thus $n\leqslant m_K$. Consider $\distinguer{t}\in V_K$ such that $\taille{\distinguer{t}}=m_K$. Since $\distinguer{t}$ is not an isolated vertex in $K$, there exists a position at which $\distinguer{t}$ is incident to a rewriting in $K$. Let $\distinguer{u}$ be the smallest prefix of this position such that $\taille{\ssterme{\distinguer{t}}{\distinguer{u}}}\leqslant M$. The term $\distinguer{t}$ can be written $\distinguer{t}=\distinguer{C}[\ssterme{\distinguer{t}}{\distinguer{u}}]$, with $\distinguer{C}$ a context.

We are going to prove that each term $t\in V_K$ is defined at position $\distinguer{u}$ and the context of $t$ at $\distinguer{u}$ is $\distinguer{C}$. It is sufficient to prove the following claim.
\begin{affirmation}\label{contexte}Let $t\in V_K$. The position $\distinguer{u}$ is prefix of every position at which the term $t$ is incident to a rewriting in $K$.
\end{affirmation}

Let $t\in \Fr_{\graph_{\GTrs{R}}}(K)$. Recall that $n\leqslant\taille{t}<n+\delta$. Since $\taille{\ssterme{t}{\distinguer{u}}}=\taille{t}-\taille{\distinguer{C}}$, we deduce
$\taille{\ssterme{t}{\distinguer{u}}}<n+\delta-\taille{\distinguer{C}}\leqslant m_K+\delta-\taille{\distinguer{C}}$. But we have $m_K-\taille{\distinguer{C}}=\taille{\ssterme{\distinguer{t}}{\distinguer{u}}}\leqslant M$.
It follows that $\taille{\ssterme{t}{\distinguer{u}}}< M+\delta$.

\end{proof}

\begin{proof}[Proof of Claim \ref{contexte}]
Suppose (as it is the case for the term $\distinguer{t}$) that there exists a position $u$ at which a term $t$ is incident to a rewriting in $K$ such that $\distinguer{u}\sqsubseteq u$
and the context of $t$ at $\distinguer{u}$ is $\distinguer{C}$. We are going to prove that if $v$ is a position at which $t$ is incident to a rewriting in $K$, then $\distinguer{u}\sqsubseteq v$. Since $K$ is connected, the claim will be proved.

First, remark that there does not exist a position $p$ smaller than $\distinguer{u}$ at which $t$ is incident to a rewriting (in $\graph_{\GTrs{R}}$): since $t\in V_K$ and $\distinguer{t}$ (which has minimal size in $V_K$) have the same context $\distinguer{C}$, we have $\taille{\ssterme{t}{\distinguer{u}}}\geqslant \taille{\ssterme{\distinguer{t}}{\distinguer{u}}}$ and thus $\taille{\ssterme{t}{p}}\geqslant \taille{\ssterme{\distinguer{t}}{p}}>M$. We deduce that if there exists a position $p$ at which $t$ is incident to a rewriting and such that $p$ and $u$ are comparable, then $\distinguer{u}\sqsubseteq p$.

Then, consider the smallest position $u_{t}$ at which the term $t$ is incident to a rewriting (Lemma \ref{reecincomp}). 
By the previous point, we have $\distinguer{u}\sqsubseteq u_{t} \sqsubseteq u$. Thus  $\distinguer{u}\sqsubseteq v$.
\end{proof}

\section{Conclusion}
We have shown that a RTL graph is word-automatic and thus its first-order theory is decidable. We have also shown that such a graph does not have a decidable $\FOAccs$ theory. Furthermore, we have shown that the concurrent unfolding of a concurrent automaton with the reachability relation is a RTL graph and therefore its $\FOAccs$ theory is decidable. Lastly, we have shown that the class of concurrent unfoldings of finite concurrent automata is not included in the class of GTR graphs since the infinite grid tree is not a GTR graph.  

Summing up, we have extended the first level of the pushdown hierarchy that consists of suffix rewriting graphs of recognizable word rewriting systems, to RTL graphs. Graphs at the first level of the pushdown hierarchy are the monadic interpretations of regular trees, that are concurrent unfoldings of finite concurrent automata for a trivial dependence relation. A RTL graph is $\FO[\REC]$ interpretation of the Cayley graph of the underlying trace monoid (Remark \ref{forec}), that is the concurrent unfolding of a finite concurrent automaton. But we do not know whether reciprocally an $\FO[\REC]$ interpretation of a concurrent unfolding of a finite concurrent automaton is a RTL graph. We do not either know whether the concurrent unfolding transformation preserves $\FOAccs$ decidability. Another interesting problem would be to extend the second level of the pushdown hierarchy.

\newpage

\nocite{*}
\bibliographystyle{eptcs}
\bibliography{sample}
\end{document}